\documentclass[copyright,creativecommons,noderivs,noncommercial]{eptcs}
\providecommand{\event}{GandALF 2016} 
\usepackage{breakurl}

\usepackage[table]{xcolor}
\usepackage{wrapfig}
\usepackage{hhline}
\usepackage{multirow}
\usepackage{paralist}
\usepackage{tikz}
\usetikzlibrary{positioning,arrows,shapes,patterns}
\usepackage{multicol}
\usepackage{mathtools}

\usepackage{algorithm}
\usepackage[noend]{algpseudocode}
\algrenewcommand{\algorithmiccomment}[1]{\hfill$\triangleleft$ {\footnotesize \textsl{#1}}}
\algloopdefx{Return}[1]{\textbf{return} #1}
\algrenewcommand\algorithmicindent{1em}
\makeatletter
\renewcommand{\ALG@beginalgorithmic}{\small}
\makeatother

\usepackage{amsmath,amsfonts,amssymb,amsthm}
\DeclareMathOperator{\PTIME}{\mathbf{P}}
\DeclareMathOperator{\NP}{\mathbf{NP}}
\DeclareMathOperator{\Psp}{\mathbf{PSPACE}}
\DeclareMathOperator{\LOGSPACE}{\mathbf{LOGSPACE}}
\DeclareMathOperator{\EXPSPACE}{\mathbf{EXPSPACE}}
\DeclareMathOperator{\NEXPTIME}{\mathbf{NEXPTIME}}

\DeclareMathOperator{\co}{\mathbf{co-}\!}
\newcommand{\Th}{\PTIME^{\NP[O(\log n)]}}
\newcommand{\Thsq}{\PTIME^{\NP[O(\log^2 n)]}}
\DeclareMathOperator{\Trk}{Trk}
\DeclareMathOperator{\Pref}{Pref}
\DeclareMathOperator{\Suff}{Suff}
\DeclareMathOperator{\states}{states}
\DeclareMathOperator{\lst}{lst}
\DeclareMathOperator{\fst}{fst}

\DeclareMathOperator{\hsA}{\langle A\rangle}
\DeclareMathOperator{\hsL}{\langle L\rangle}
\DeclareMathOperator{\hsB}{\langle B\rangle}
\DeclareMathOperator{\hsE}{\langle E\rangle}
\DeclareMathOperator{\hsD}{\langle D\rangle}
\DeclareMathOperator{\hsO}{\langle O\rangle}

\DeclareMathOperator{\hsAt}{\langle \overline{A}\rangle}
\DeclareMathOperator{\hsLt}{\langle \overline{L}\rangle}
\DeclareMathOperator{\hsBt}{\langle \overline{B}\rangle}
\DeclareMathOperator{\hsEt}{\langle \overline{E}\rangle}
\DeclareMathOperator{\hsDt}{\langle \overline{D}\rangle}
\DeclareMathOperator{\hsOt}{\langle \overline{O}\rangle}

\newcommand{\A}{\mathsf{A}}
\newcommand{\Abar}{\mathsf{\overline{A}}}
\newcommand{\AAbar}{\mathsf{A\overline{A}}}
\newcommand{\AAbarB}{\mathsf{A\overline{A}B}}
\newcommand{\AAbarE}{\mathsf{A\overline{A}E}}
\newcommand{\AB}{\mathsf{AB}}
\newcommand{\AbarE}{\mathsf{\overline{A}E}}
\newcommand{\AbarB}{\mathsf{\overline{A}B}}
\renewcommand{\AE}{\mathsf{AE}}
\newcommand{\AAbarBbarEbar}{\mathsf{A\overline{A}\overline{B}\overline{E}}}
\newcommand{\AAbarBBbar}{\mathsf{A\overline{A}B\overline{B}}}
\newcommand{\AAbarEEbar}{\mathsf{A\overline{A}E\overline{E}}}

\newcommand{\Bbar}{\mathsf{\overline{B}}}
\newcommand{\Ebar}{\mathsf{\overline{E}}}
\newcommand{\B}{\mathsf{B}}
\newcommand{\E}{\mathsf{E}}
\newcommand{\BE}{\mathsf{BE}}

\newcommand{\AAbarBBbarEbar}{\mathsf{A\overline{A}B\overline{B}\overline{E}}}
\newcommand{\AAbarEBbarEbar}{\mathsf{A\overline{A}E\overline{B}\overline{E}}}

\newcommand{\HSprop}{\mathsf{Prop}}

\DeclareMathOperator{\mods}{ModSubf_\AAbar}

\DeclareMathAlphabet{\mathpzc}{OT1}{pzc}{m}{it}

\newcommand{\Ku}{\ensuremath{\mathpzc{K}}}

\newtheorem{definition}{Definition}
\newtheorem{theorem}{Theorem}

\newtheorem{lemma}{Lemma}
\newtheorem{corollary}{Corollary}

\title{Model Checking the Logic of Allen's Relations\\ \emph{Meets} and \emph{Started-by} is $\PTIME^{\NP}$-Complete}

\author{Laura Bozzelli
\institute{Technical University of Madrid (UPM), Madrid, Spain\hspace*{-0.9cm}}
\email{laura.bozzelli@fi.upm.es}
\and
Alberto Molinari \qquad Angelo Montanari
\institute{University of Udine, Udine, Italy}
\email{molinari.alberto@gmail.com angelo.montanari@uniud.it}
\and
Adriano Peron
\institute{University of Napoli ``Federico II'', Napoli, Italy}
\email{adrperon@unina.it}
\and
Pietro Sala
\institute{University of Verona, Verona, Italy}
\email{pietro.sala@univr.it}
}
\def\titlerunning{Model Checking the Logic of Allen's Relations \emph{Meets} and \emph{Started-by} is $\PTIME^{\NP}$-Complete}
\def\authorrunning{L. Bozzelli, A. Molinari, A. Montanari, A. Peron \& P. Sala}
\def\copyrightholders{Bozzelli, Molinari, Montanari, Peron \& Sala}

\begin{document}
\maketitle

\begin{abstract}
In the plethora of fragments of Halpern and Shoham's modal logic of time intervals (HS), the logic $\AB$ of Allen's relations \emph{Meets} and \emph{Started-by} is at a central position. Statements that may be true at certain intervals, but at no sub-interval of them, such as accomplishments, as well as metric constraints about the length of intervals, that force, for instance, an interval to be at least (resp., at most, exactly) $k$ points long, can be expressed in $\AB$. Moreover, over the linear order of the natural numbers $\mathbb{N}$, it subsumes the (point-based) logic LTL, as it can easily encode the next and until modalities. Finally, it is expressive enough to capture the $\omega$-regular languages, that is, for each $\omega$-regular expression $R$ there exists an $\AB$ formula $\varphi$ such that the language defined by $R$ coincides with the set of models of $\varphi$ over $\mathbb{N}$. It has been shown that the satisfiability problem for $\AB$ over $\mathbb{N}$ is $\EXPSPACE$-complete. Here we prove that, under the homogeneity assumption, its model checking problem is $\Delta^p_2 = \PTIME^{\NP}$-complete (for the sake of comparison, the model checking problem for full HS is $\EXPSPACE$-hard, and the only known decision procedure is nonelementary). Moreover, we show that the modality for the Allen relation \emph{Met-by} can be added to $\AB$ at no extra cost ($\AAbarB$ is $\PTIME^{\NP}$-complete as well).

%
\end{abstract}

\section{Introduction}

In this paper, we investigate the model checking problem for the interval logic of Allen's Relations \emph{Meets} and \emph{Started-by}.
Given a model of a system (generally, a Kripke structure) and a temporal logic formula, which specifies the expected properties of the system,
model checking algorithms verify, in fully automatic way, whether the model satisfies the formula; if this is not the case, they provide a 
counterexample, that is, a computation of the system failing to satisfy some property. Model checking has been successfully employed in formal 
verification as well as in various areas of AI, ranging from planning to configuration and multi-agent systems~\cite{DBLP:conf/ecp/GiunchigliaT99,DBLP:conf/tacas/LomuscioR06}.

Standard point-based temporal logics, such as LTL, CTL, and CTL$^{*}$~\cite{emerson1986sometimes,pnueli1977temporal}, are commonly used as 
specification languages. Even though they turn out to be well-suited for a variety of application domains, there are relevant system 
properties, involving, for instance, actions with duration, accomplishments, and temporal aggregations, which are inherently ``interval-based'' 
and thus cannot be properly dealt with by temporal logics that allow one to predicate over computation states only.
To overcome these limitations,
one can resort to \emph{interval temporal logics} (ITLs), that take intervals---instead of points---as their primitive entities~\cite{HS91}, which have been successfully applied in various areas of computer science and AI, including hardware and software verification, computational linguistics, and planning \cite{LM13,digital_circuits_thesis,DBLP:journals/ai/Pratt-Hartmann05,DBLP:series/eatcs/ChaochenH04}. 

ITL model checking is the context of this paper. In order to check interval properties of computations, one needs to collect information about states into computation stretches: each finite path of a Kripke structure is interpreted as an interval, whose labelling is defined on the basis of the labelling of the component states.
Among ITLs, Halpern and Shoham's modal logic of time intervals HS~\cite{HS91} is the main reference. It features one modality for each possible ordering relation between a pair of intervals apart from equality (the so-called Allen's relations~\cite{All83}).
The \emph{satisfiability problem} for HS has been thoroughly studied, and it turns out to be highly undecidable for all relevant (classes of) linear orders~\cite{HS91}. The same holds for most HS fragments~\cite{DBLP:journals/amai/BresolinMGMS14}; however, some meaningful exceptions exist, including the logic of temporal neighbourhood $\AAbar$ and the logic of sub-intervals $\mathsf{D}$~\cite{BGMS10,BGMS09}.
The \emph{model checking problem} for HS has entered the research agenda only recently~\cite{bmmps16,LM13,LM14,LM16,MMMPP15,MMP15B,MMP15,MMPS16}.
In \cite{MMMPP15}, Molinari et al.\ deal with model checking for full HS over Kripke structures under the homogeneity assumption~\cite{Roe80},
showing its non-elementary decidability by means of a suitable small model theorem ($\EXPSPACE$-hardness has been proved in \cite{bmmps16}).
%
Since then, the attention was brought to HS fragments, which, similarly to what happens with satisfiability, are often computationally better. 

In this paper we first prove that model checking for the logic $\AAbarB$ (resp., $\AAbarE$) of Allen's relations \emph{Meets}, \emph{Met-by}, and \emph{Started-by} (resp., \emph{Finished-by}) is in $\PTIME^{\NP}$; then we prove that its fragment $\AB$ (resp., $\AbarE$) is $\PTIME^{\NP}$-hard; finally we show that its fragment $\AbarB$ (resp., $\AE$) belongs to $\Thsq$ and it is $\Th$-hard.
%
$\PTIME^{\NP}$ (also denoted as $\Delta^p_2$) is the class of problems decided by a deterministic polynomial time Turing machine that queries an NP oracle.
The classes $\Th$ and $\Thsq$ are analogous, but the number of queries is bounded by $O(\log n)$ and $O(\log^2 n)$, respectively, being $n$ the input size~\cite{gottlob1995,schnoebelen2003}. 
These three classes are higher than both $\NP$ and $\co\NP$ in the \emph{polynomial time hierarchy}, and closed under complement. It is worth noticing that, whereas we know many natural problems which are complete for $\Sigma^p_2$ or $\Pi^p_2$ (in general, for $\Sigma^p_k$ and $\Pi^p_k$, with $k\geq 2$), the classes $\PTIME^{\NP}$, $\Th$, and $\Thsq$ are not so ``populated'' (and neither are the classes $\Delta^p_k$, for $k > 2$). 
Among the few natural problems complete for $\PTIME^{\NP}$, we would like to mention model checking for several fragments of CTL$^*$, including CTL$^+$, ECTL$^+$, and FCTL \cite{LMS01}. As for the other two classes, very recently Molinari et al.\ have shown that model checking $\A$, $\Abar$, or $\AAbar$ formulas is in $\Thsq$ and hard for $\Th$ \cite{MMPS16}.

\smallskip

\noindent \textbf{Related work.}
In~\cite{LM13,LM14}, Lomuscio and Michaliszyn address the model checking problem for some fragments of HS extended with epistemic modalities. Their semantic assumptions considerably differ from those made in \cite{MMMPP15}, making it difficult to compare the outcomes of the two research lines. 
Moreover, they consider a restricted form of model checking, which verifies a specification against a single (finite) initial computation interval (this is in general a limitation, unless some operators of HS are available, such as $\hsA$): their goal is indeed to reason about a given computation of a multi-agent system, rather than on all its admissible computations.
Recently they have shown how to exploit regular expressions in order to specify the way in which the intervals of a Kripke structure get labelled~\cite{LM16}. Such an extension leads to a significant increase in 
the expressiveness of HS formulas.

\smallskip

\noindent \textbf{Organization of the paper.}
In the next section we introduce the fundamental elements of the model checking problem for HS and its fragments. Then, in Section~\ref{sec:AAbarBalgo}, we provide a $\PTIME^{\NP}$ model checking algorithm for $\AAbarB$ (and $\AAbarE$) formulas. In Section~\ref{sec:ABhard} we prove the $\PTIME^{\NP}$-hardness of model checking for $\AB$ (and $\AbarE$). 
$\PTIME^{\NP}$-completeness of $\AB$, $\AbarE$, $\AAbarB$ and $\AAbarE$ follows.
Finally we show that the problem for formulas of $\AbarB$ and $\AE$ is in $\Thsq$ and hard for $\Th$.

\section{Preliminaries}\label{sec:backgr}


\textbf{The interval temporal logic HS.}
An interval algebra to reason about intervals and their relative order was proposed by Allen in~\cite{All83}, while a systematic logical study of interval representation and reasoning was done a few years later by Halpern and Shoham, who introduced the interval temporal logic HS featuring one modality for each Allen
relation, but equality~\cite{HS91}.
Table~\ref{allen} depicts 6 of the 13 Allen's relations,
together with the corresponding HS (existential) modalities. The other 7 relations are the 6 inverses (given a binary relation $\mathpzc{R}$, the inverse $\overline{\mathpzc{R}}$ is such that $b \overline{\mathpzc{R}} a$ if and only if $a \mathpzc{R} b$) and equality. 

\begin{table}[tb]
\centering
\caption{Allen's relations and corresponding HS modalities.}\label{allen}
\vspace*{0.1cm}
\resizebox{\width}{0.8\height}{
\begin{tabular}{cclc}
\hline
\rule[-1ex]{0pt}{3.5ex} Allen relation & HS & Definition w.r.t. interval structures &  Example\\ 
\hline

&   &   & \multirow{7}{*}{\begin{tikzpicture}[scale=0.96]
\draw[draw=none,use as bounding box](-0.3,0.2) rectangle (3.3,-3.1);
\coordinate [label=left:\textcolor{red}{$x$}] (A0) at (0,0);
\coordinate [label=right:\textcolor{red}{$y$}] (B0) at (1.5,0);
\draw[red] (A0) -- (B0);
\fill [red] (A0) circle (2pt);
\fill [red] (B0) circle (2pt);

\coordinate [label=left:$v$] (A) at (1.5,-0.5);
\coordinate [label=right:$z$] (B) at (2.5,-0.5);
\draw[black] (A) -- (B);
\fill [black] (A) circle (2pt);
\fill [black] (B) circle (2pt);

\coordinate [label=left:$v$] (A) at (2,-1);
\coordinate [label=right:$z$] (B) at (3,-1);
\draw[black] (A) -- (B);
\fill [black] (A) circle (2pt);
\fill [black] (B) circle (2pt);

\coordinate [label=left:$v$] (A) at (0,-1.5);
\coordinate [label=right:$z$] (B) at (1,-1.5);
\draw[black] (A) -- (B);
\fill [black] (A) circle (2pt);
\fill [black] (B) circle (2pt);

\coordinate [label=left:$v$] (A) at (0.5,-2);
\coordinate [label=right:$z$] (B) at (1.5,-2);
\draw[black] (A) -- (B);
\fill [black] (A) circle (2pt);
\fill [black] (B) circle (2pt);

\coordinate [label=left:$v$] (A) at (0.5,-2.5);
\coordinate [label=right:$z$] (B) at (1,-2.5);
\draw[black] (A) -- (B);
\fill [black] (A) circle (2pt);
\fill [black] (B) circle (2pt);

\coordinate [label=left:$v$] (A) at (1.3,-3);
\coordinate [label=right:$z$] (B) at (2.3,-3);
\draw[black] (A) -- (B);
\fill [black] (A) circle (2pt);
\fill [black] (B) circle (2pt);

\coordinate (A1) at (0,-3);
\coordinate (B1) at (1.5,-3);
\draw[dotted] (A0) -- (A1);
\draw[dotted] (B0) -- (B1);
\end{tikzpicture}}\\ 

\textsc{meets} & $\hsA$ & $[x,y]\mathpzc{R}_A[v,z]\iff y=v$ &\\ 

\textsc{before} & $\hsL$ & $[x,y]\mathpzc{R}_L[v,z]\iff y<v$ &\\ 
 
\textsc{started-by} & $\hsB$ & $[x,y]\mathpzc{R}_B[v,z]\iff x=v\wedge z<y$ &\\ 

\textsc{finished-by} & $\hsE$ & $[x,y]\mathpzc{R}_E[v,z]\iff y=z\wedge x<v$ &\\ 

\textsc{contains} & $\hsD$ & $[x,y]\mathpzc{R}_D[v,z]\iff x<v\wedge z<y$ &\\ 

\textsc{overlaps} & $\hsO$ & $[x,y]\mathpzc{R}_O[v,z]\iff x<v<y<z$ &\\

\hline
\end{tabular}}
\end{table}

The HS language consists of a set of proposition letters $\mathpzc{AP}$, the Boolean connectives $\neg$ and $\wedge$, 
and a temporal modality for each of the (non trivial) Allen's relations, i.e., $\hsA$, $\hsL$, $\hsB$, $\hsE$, $\hsD$, $\hsO$, $\hsAt$, $\hsLt$, $\hsBt$, $\hsEt$, $\hsDt$, and $\hsOt$.
%
HS formulas are defined by the grammar
$
    \psi ::= p \;\vert\; \neg\psi \;\vert\; \psi \wedge \psi \;\vert\; \langle X\rangle\psi \;\vert\; \langle \overline{X}\rangle\psi,
$
where $p\in\mathpzc{AP}$ and $X\in\{A,L,B,E,D,O\}$.
In the following, we shall also exploit as abbreviations the standard logical connectives for disjunction $\vee$, implication $\rightarrow$, and double implication $\leftrightarrow$.
%
Furthermore, for any modality $X$, the dual universal modalities $[X]\psi$ and $[\overline{X}]\psi$ are defined as $\neg\langle X\rangle\neg\psi$ and $\neg\langle \overline{X} \rangle\neg\psi$, respectively. 
Finally, given any subset of Allen's relations $\{X_1,\cdots,X_n\}$, we denote by $\mathsf{X_1 \cdots X_n}$ the HS fragment featuring existential (and universal) modalities for $X_1,\ldots, X_n$ only. 

    
W.l.o.g., we assume the \emph{non-strict semantics of HS}, which admits intervals consisting of a single point\footnote{All the results we prove in the paper hold for the strict semantics as well.}. Under such an assumption, all HS modalities can be expressed in terms of modalities 
$\hsB, \hsE, \hsBt$, and $\hsEt$~\cite{HS91}.
%
HS can thus be seen as a multi-modal logic with these 
$4$ primitive modalities
and its semantics can be defined over a multi-modal Kripke structure, called \emph{abstract interval model}, where intervals are treated as atomic objects and Allen's relations as 
binary relations between pairs of intervals.
Since later we will focus on some HS fragments excluding $\hsBt$ and $\hsEt$, we add both $\hsA$ and $\hsAt$ to the considered set of HS modalities.

\begin{definition}\cite{MMMPP15}
An \emph{abstract interval model} is a tuple $\mathpzc{A}=(\mathpzc{AP},\mathbb{I},A_\mathbb{I},B_\mathbb{I},E_\mathbb{I},\sigma)$, where
     $\mathpzc{AP}$ is a set of proposition letters,
     $\mathbb{I}$ is a possibly infinite set of atomic objects (worlds),
     $A_\mathbb{I}$, $B_\mathbb{I}$, and $E_\mathbb{I}$ are three binary relations over $\mathbb{I}$, and
     $\sigma:\mathbb{I}\mapsto 2^{\mathpzc{AP}}$ is a (total) labeling function, which assigns a set of proposition letters to each world.
\end{definition}
In the interval setting, $\mathbb{I}$ is interpreted as a set of intervals and $A_\mathbb{I}$, $B_\mathbb{I}$, and $E_\mathbb{I}$ as Allen's relations $A$ (\emph{meets}), $B$ (\emph{started-by}), and $E$ (\emph{finished-by}), respectively; $\sigma$ assigns to each interval in $\mathbb{I}$ the set of proposition letters that hold over it.

Given an abstract interval model $\mathpzc{A}=(\mathpzc{AP},\mathbb{I}, A_\mathbb{I}, B_\mathbb{I},E_\mathbb{I},\sigma)$ and an interval $I\in\mathbb{I}$, the truth of an HS formula over $I$ is inductively defined as follows:
\begin{compactitem}
    \item $\mathpzc{A},I\models p$ iff $p\in \sigma(I)$, for any $p\in\mathpzc{AP}$;
    \item $\mathpzc{A},I\models \neg\psi$ iff it is not true that $\mathpzc{A},I\models \psi$ (also denoted as $\mathpzc{A},I\not\models \psi$);
    \item $\mathpzc{A},I\models \psi \wedge \phi$ iff $\mathpzc{A},I\models \psi$ and $\mathpzc{A},I\models \phi$;
    \item $\mathpzc{A},I\models \langle X\rangle\psi$, for $X \in\{A,B,E\}$, iff there exists $J\in\mathbb{I}$ such that $I\, X_\mathbb{I}\, J$ and $\mathpzc{A},J\models \psi$;
    \item $\mathpzc{A},I\models \langle \overline{X}\rangle\psi$, for $\overline{X} \in\{\overline{A},\overline{B},\overline{E}\}$, iff there exists $J\in\mathbb{I}$ such that $J\, X_\mathbb{I}\, I$ and $\mathpzc{A},J\models \psi$.
\end{compactitem}

\smallskip

\noindent \textbf{Kripke structures and abstract interval models.}
In model checking, finite state systems are usually modelled as Kripke structures. 
In \cite{MMMPP15}, the authors define a mapping
from Kripke structures to abstract interval models, that allows one
to specify interval properties of computations by means of HS formulas. 

\begin{definition}
A \emph{finite Kripke structure} is a tuple $\mathpzc{K}=(\mathpzc{AP},W, \delta,\mu,w_0)$, where $\mathpzc{AP}$ is a set of proposition letters, $W$ is a finite set of states, 
$\delta\subseteq W\times W$ is a left-total relation between pairs of states,
$\mu:W\mapsto 2^\mathpzc{AP}$ is a total labelling function, and $w_0\in W$ is the initial state.
\end{definition}

For all $w\in W$, $\mu(w)$ is the set of proposition letters that hold at $w$,
while $\delta$ is the transition relation that describes the evolution of the system over time.


\begin{wrapfigure}[4]{l}{0.36\linewidth}
\vspace*{-0.4cm}
\centering
\begin{tikzpicture}[->,>=stealth,thick,shorten >=1pt,auto,node distance=2cm,every node/.style={circle,draw}]
    \node [style={double}](v0) {$\stackrel{v_0}{p}$};
    \node (v1) [right of=v0] {$\stackrel{v_1}{q}$};
    \draw (v0) to [bend right] (v1);
    \draw (v1) to [bend right] (v0);
    \draw (v0) to [loop left] (v0);
    \draw (v1) to [loop right] (v1);
\end{tikzpicture}
\vspace*{-0.15cm}
\caption{The Kripke structure $\mathpzc{K}_2$.}\label{KEquiv}
\end{wrapfigure}
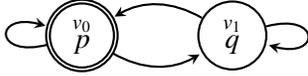

Figure~\ref{KEquiv} depicts the finite Kripke structure $\mathpzc{K}_2 = (\{p,q\},\allowbreak \{v_0,v_1\},\delta,\mu,v_0)$,
where 
$\delta=\{(v_0,v_0),(v_0,v_1),(v_1,v_0),(v_1,v_1)\}$,
$\mu(v_0)\!=\!\{p\}$, and $\mu(v_1)\!=\!\{q\}$. 
The initial state $v_0$ is identified by a double circle.

\begin{definition}
A \emph{track} $\rho$ over a finite Kripke structure $\mathpzc{K}=(\mathpzc{AP},W,\delta,\mu,w_0)$ is a finite sequence of states $v_1\cdots v_n$, with $n\geq 1$, such that $(v_i,v_{i+1})\in \delta$ for $i = 1,\ldots ,n-1$.
\end{definition}

\noindent
Let $\Trk_\mathpzc{K}$ be the (infinite) set of all tracks over a finite Kripke structure $\mathpzc{K}$. For any track $\rho=v_1\cdots v_n \in \Trk_\mathpzc{K}$, we define:
\begin{compactitem}
\item $|\rho|=n$, $\fst(\rho)=v_1$, $\lst(\rho)=v_n$, and for $1\leq i\leq |\rho|$, $\rho(i)=v_i$;
\item $\states(\rho)=\{v_1,\cdots,v_n\}\subseteq W$;
\item $\rho(i,j)=v_i\cdots v_j$, with $1\leq i \leq j\leq |\rho|$, is the subtrack of $\rho$ bounded by $i$ and $j$;
\item $\Pref(\rho)=\{\rho(1,i) \mid 1\leq i\leq |\rho|-1\}$ and $\Suff(\rho)=\{\rho(i,|\rho|) \mid 2\leq i\leq |\rho|\}$ are the sets of all proper prefixes and suffixes of $\rho$, respectively.
\end{compactitem}
%
Finally, if $\fst(\rho)=w_0$ (the initial state of $\mathpzc{K}$), 
$\rho$ is called an \emph{initial track}. 

An abstract interval model (over $\Trk_\mathpzc{K}$) can be naturally associated with a finite Kripke structure $\mathpzc{K}$ by considering the set of intervals as the set of tracks of $\mathpzc{K}$. Since $\mathpzc{K}$ has loops ($\delta$ is left-total), the number of tracks in $\Trk_\mathpzc{K}$, and thus the number of intervals, 
is infinite.

\begin{definition}\label{def:inducedmodel}
The \emph{abstract interval model induced by a finite Kripke structure} $\mathpzc{K}=(\mathpzc{AP},W,\delta,\mu,w_0)$ is  
$\mathpzc{A}_\mathpzc{K}=(\mathpzc{AP},\mathbb{I},A_\mathbb{I},B_\mathbb{I},E_\mathbb{I},\sigma)$, where
    $\mathbb{I}=\Trk_\mathpzc{K}$,
    $A_\mathbb{I}=\{(\rho,\rho')\in\mathbb{I}\times\mathbb{I}\mid \lst(\rho)=\fst(\rho')\}$,
    $B_\mathbb{I}=\{(\rho,\rho')\in\mathbb{I}\times\mathbb{I}\mid \rho'\in\Pref(\rho)\}$,
    $E_\mathbb{I}=\{(\rho,\rho')\in\mathbb{I}\times\mathbb{I}\mid \rho'\in\Suff(\rho)\}$, and
    $\sigma:\mathbb{I}\mapsto 2^\mathpzc{AP}$ is such that $\sigma(\rho)=\bigcap_{w\in\states(\rho)}\mu(w)$, for all $\rho\in\mathbb{I}$.
\end{definition}
\noindent
Relations $A_\mathbb{I},B_\mathbb{I}$, and $E_\mathbb{I}$ are interpreted as the Allen's relations $A,B$, and $E$, respectively. Moreover, according to the definition of $\sigma$, 
$p\in\mathpzc{AP}$ holds over $\rho=v_1\cdots v_n$ 
iff it holds over all the states $v_1, \cdots , v_n$ of $\rho$. This conforms to the \emph{homogeneity principle}, according to which a proposition letter holds over an interval 
if and only if it holds over all its subintervals~\cite{Roe80}.

\begin{definition}
Let $\mathpzc{K}$ be a finite Kripke structure and
$\psi$ be an HS formula; 
we say that a track $\rho\in\Trk_\mathpzc{K}$ satisfies $\psi$,
denoted as $\mathpzc{K},\rho\models \psi$, iff it holds that $\mathpzc{A}_\mathpzc{K},\rho\models \psi$.
%
Moreover, we say that
$\mathpzc{K}$ models $\psi$, denoted as $\mathpzc{K}\models \psi$, iff 
for all \emph{initial} tracks $\rho'\in\Trk_\mathpzc{K}$ it holds that $\mathpzc{K},\rho'\models \psi$.
The \emph{model checking problem} for HS over finite Kripke structures is 
the problem of deciding whether $\mathpzc{K}\models \psi$.
\end{definition}

We conclude with a simple example (a simplified version of the one given in \cite{MMMPP15}), showing that the fragments investigated in this paper can express meaningful properties of state transition systems.


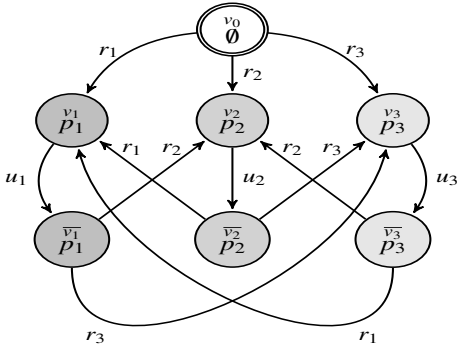
\begin{wrapfigure}[12]{L}{0.40\linewidth}
\vspace*{-0.3cm}
\centering
\resizebox{0.97\width}{0.72\height}{
\begin{tikzpicture}[->,>=stealth',shorten >=1pt,auto,node distance=2.2cm,thick,main node/.style={circle,draw}]

  \node[main node,style={double}] (1) {$\stackrel{v_0}{\emptyset}$};
  \node[main node,fill=gray!35] (3) [below=0.7cm of 1] {$\stackrel{v_2}{p_2}$};
  \node[main node,fill=gray!50] (2) [left of=3] {$\stackrel{v_1}{p_1}$};
  \node[main node,fill=gray!20] (4) [right of=3] {$\stackrel{v_3}{p_3}$};
  \node[main node,fill=gray!50] (5) [below of=2] {$\stackrel{\overline{v_1}}{p_1}$};
  \node[main node,fill=gray!35] (6) [below of=3] {$\stackrel{\overline{v_2}}{p_2}$};
  \node[main node,fill=gray!20] (7) [below of=4] {$\stackrel{\overline{v_3}}{p_3}$};

  \path[every node/.style={font=\small}]
    (1) edge [bend right] node[left] {$r_1$} (2)
        edge node {$r_2$} (3)
        edge [bend left] node[right] {$r_3$} (4)
    (2) edge [bend right] node [left] {$u_1$} (5)
    (3) edge node {$u_2$} (6)
    (4) edge [bend left] node [right] {$u_3$} (7)
    (5) edge node[very near end,left] {$r_2$} (3)
    (5) edge [out=270,in=260,looseness=1.3] node [near start,swap] {$r_3$} (4)
    (6) edge node[very near end,right] {$r_1$} (2)
    (6) edge node[very near end,left] {$r_3$} (4)
    (7) edge [out=270,in=280,looseness=1.3] node [near start] {$r_1$} (2)
    (7) edge node[very near end,right] {$r_2$} (3)
    ;
\end{tikzpicture}}
\vspace{-1.45cm}
\caption{The Kripke structure $\mathpzc{K}_{Sched}$.}\label{KSched}
\end{wrapfigure}

In Figure~\ref{KSched}, we provide an example of a finite Kripke structure $\mathpzc{K}_{Sched}$ that models the behaviour of a scheduler serving three processes which are continuously requesting the use of a common resource. The initial state 
is $v_0$: no process is served in that state. In the states $v_i$ and $\overline{v}_i$, with $i \in \{1,2,3\}$, the $i$-th process is served (this is denoted by the fact that $p_i$ holds in those states). For the sake of readability, edges are marked either by $r_i$, for $request(i)$, or by $u_i$, for $unlock(i)$. Edge labels do not have a semantic value, that is, they are neither part of the structure definition, nor proposition letters; they are simply used to ease reference to edges. 
Process $i$ is served in state $v_i$, then, after ``some time'', a transition $u_i$ from $v_i$ to $\overline{v}_i$ is taken; subsequently, process $i$ cannot be served again immediately, as $v_i$ is not directly reachable from $\overline{v}_i$ (the scheduler cannot serve the same process twice in two successive rounds). A transition $r_j$, with $j\neq i$, from $\overline{v}_i$ to $v_j$ is then taken and process $j$ is served. This structure can easily be generalised to a higher number of processes.

We now show how some meaningful properties to be checked over 
$\mathpzc{K}_{Sched}$ 
can be expressed in the HS fragment $\AbarE$. 
In all the following formulas, we force the validity of the considered properties over all legal computation sub-intervals by using the modality $[E]$ (all computation sub-intervals are suffixes of at least one initial track of the Kripke structure). The first formula requires that at least 2 proposition letters are witnessed in any suffix of length at least 4 of an initial track. Since a process cannot be executed twice in a row, it is satisfied by $\mathpzc{K}_{Sched}$:
$\mathpzc{K}_{Sched}\models[E]\big(\hsE^3\top \rightarrow (\chi(p_1,p_2) \vee \chi(p_1,p_3) \vee \chi(p_2,p_3))\big)$ where $\chi(p,q)\!=\!\hsE\hsAt p \wedge \hsE\hsAt q$.
The second formula requires that, in any suffix of length at least 11 of an initial track, process 3 is executed at least once in some internal states (\emph{non starvation}). $\mathpzc{K}_{Sched}$ does not satisfy it, because the scheduler can postpone the execution of a process ad libitum: $\mathpzc{K}_{Sched}\not\models[E](\hsE^{10}\top \rightarrow \hsE\hsAt p_3)$. 
The third formula requires that, in any suffix of length at least 6 of an initial track, $p_1,p_2$, and $p_3$ are all witnessed. The only way to satisfy this property would be to force the scheduler to execute the three processes in a strictly periodic manner (\emph{strict alternation}), that is, $p_ip_jp_kp_ip_jp_kp_ip_jp_k\cdots$, for $i,j,k \in \{1,2,3\}$ and $i \neq  j \neq  k \neq i$, but $\mathpzc{K}_{Sched}$ does not meet such a
requirement: $\mathpzc{K}_{Sched}\not\models[E](\hsE^5 \rightarrow (\hsE\hsAt p_1 \wedge \hsE\hsAt p_2 \wedge \hsE\hsAt p_3))$.
%

\smallskip

\noindent\textbf{The general picture.}
We now describe known and new complexity results about the model checking problem for HS fragments (see Figure~\ref{fig:overv} for a graphical account).
%

\begin{figure}[tb]
\centering
\resizebox{0.8\width}{0.55\height}{
\newcommand{\cellThree}[3]{
\begin{tabular}{c|c}
\rule[-1ex]{0pt}{3.5ex}
\multirow{2}{*}{#1} & #2 \\ 
\hhline{~=}\rule[-1ex]{0pt}{3.5ex}
 & #3 
\end{tabular}}

\newcommand{\cellTwo}[2]{\begin{tabular}{c|c}
\rule[-1ex]{0pt}{3.5ex}
#1 & #2 \\
\end{tabular}}

\begin{tikzpicture}[->,>=stealth',shorten >=1pt,auto,semithick,main node/.style={rectangle,draw, inner sep=0pt}]  

\tikzstyle{gray node}=[fill=gray!30]

    \node [main node](0) at (-4,0) {\cellTwo{$\AAbarBbarEbar$}{PSPACE-complete $^{2,3}$}};
    \node [main node](1) at (3,0)  {\cellTwo{$\Bbar$}{$\Psp$-complete $^4$}};
    \node [main node](21) at (3,0.8)  {\cellTwo{$\Ebar$}{$\Psp$-complete $^4$}};
    \node [main node](31) at (3,2.3)  {\cellTwo{$\AAbarEEbar$}{$\Psp$-complete $^5$}};
    \node [main node](32) at (-2.3,2.3)  {\cellTwo{$\AAbarBBbar$}{$\Psp$-complete $^5$}};
    \node [main node](2) at (-6.5,-3.4) {\cellThree{$\AAbar$}{$\Thsq$ $^4$}{$\Th$-hard $^4$}};
    \node [main node](22) at (0,-3.4) {\cellThree{$\A$, $\Abar$}{$\Thsq$ $^4$}{$\Th$-hard $^4$}};
    \node [main node](92) at (6.5,-3.4) {\cellThree{$\AbarB$, $\AE$}{\cellcolor{gray!30}{$\Thsq$}}{\cellcolor{gray!30}{$\Th$-hard}}};
    
    \node [main node](72) at (-3.3,-1.0)  {\cellTwo{$\AAbarB$}{\cellcolor{gray!30}{$\PTIME^{\NP}$-complete}}};
    \node [main node](73) at (3.5,-1.0)  {\cellTwo{$\AAbarE$}{\cellcolor{gray!30}{$\PTIME^{\NP}$-complete}}};
    \node [main node](74) at (-3.3,-2)  {\cellTwo{$\AB$}{\cellcolor{gray!30}{$\PTIME^{\NP}$-complete}}};
    \node [main node](75) at (3.5,-2)  {\cellTwo{$\AbarE$}{\cellcolor{gray!30}{$\PTIME^{\NP}$-complete}}};
    
    \node [main node](53) at (-4,-6) {\cellTwo{$\B$}{coNP-complete $^5$}};
    \node [main node](3) at (-4,-5) {\cellTwo{$\E$}{coNP-complete $^5$}};
    \node [main node](4) at (3.5,-5.5) {\cellTwo{$\HSprop$}{coNP-complete $^3$}};
    
    \node [main node](5) at (-4,4) {\cellThree{$\AAbarBBbarEbar$}{EXPSPACE $^2$}{PSPACE-hard $^3$}};
    \node [main node](6) at (-3.5,6) {\cellThree{succinct $\AAbarBBbarEbar$}{EXPSPACE $^2$}{NEXP-hard $^2$}};
    
    \node [main node](9) at (3,6) {\cellThree{$\BE$}{nonELEMENTARY $^1$}{EXPSPACE-hard $^5$}};
    \node [main node](10) at (3,8) {\cellThree{full HS}{nonELEMENTARY $^1$}{EXPSPACE-hard $^5$}};
   
    \path
    (1) edge [swap] node {hardness} (0) 
    (0) edge  [out=150,in=190] node {hardness} (5.south)
    (4.west) edge [swap,near end] node {hardness} (3.east)
    (4.west) edge [near end] node {hardness} (53.east)
    (2.north east) edge [swap,out=370,in=170] node {upper-bound} (22.north west)
    (22.south west) edge [swap,out=190,in=-10] node {hardness} (2.south east)
    (22.east) edge node {hardness} (92.west)
    (9) edge [swap] node {hardness} (10)
    (21.north) edge node {hardness} (31.south)
    (1.west) edge [out=120,in=-45, near end] node {hardness} (32.south)
    (74.west) edge [out=130, in=230, near end] node {hardness} (72.west)
    (72.east) edge [out=310, in=50] node {upper-bound} (74.east)
    (75.west) edge [out=130, in=230, near end] node {hardness} (73.west)
    (73.east) edge [out=310, in=50] node {upper-bound} (75.east)
    ;
    
    \draw [dashed,-,gray] (-6.5,4) -- (6,4);
    \draw [dashed,-,gray] (-6.5,-2.5) -- (6,-2.5);
    \draw [dashed,-,gray] (-6.5,-4.5) -- (6,-4.5);
    \draw [dashed,-,gray] (-6.5,-0.5) -- (6,-0.5);
    
    \node[align=left](50) at (4,-7) {
    $^1$ \cite{MMMPP15}, 
    $^2$ \cite{MMP15}, 
    $^3$ \cite{MMP15B}, 
    $^4$ \cite{MMPS16},
    $^5$ \cite{bmmps16}};

\end{tikzpicture}}
\vspace{-0.2cm}
\caption{Complexity of the model checking problem for HS fragments: known results are depicted in white boxes, new ones in gray boxes.}\label{fig:overv}
\end{figure}
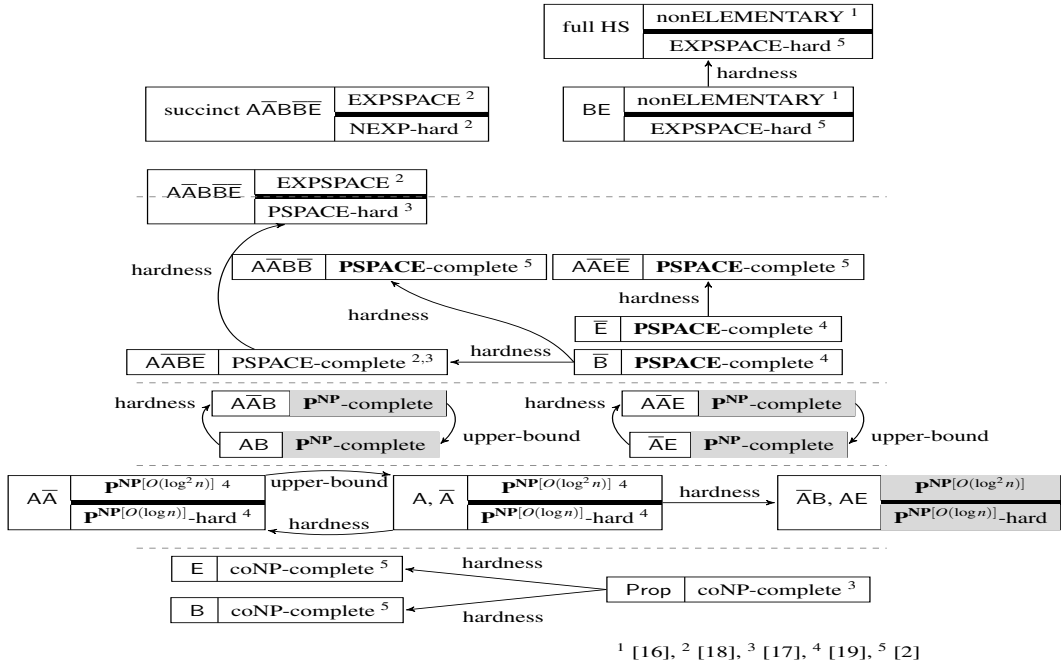

In \cite{MMMPP15}, Molinari et al.\ have shown that, given a Kripke structure $\mathpzc{K}$ and a bound $k$ on the structural complexity of HS formulas, i.e., on the nesting depth of $\hsE$ and $\hsB$ modalities, it is possible to obtain a \emph{finite} representation for $\mathpzc{A}_\mathpzc{K}$, which is equivalent to $\mathpzc{A}_\mathpzc{K}$ with respect to satisfiability of HS formulas with structural complexity less than or equal to $k$. Then, by exploiting such a representation, they proved that the model checking problem for (full) HS is decidable, providing an algorithm with non-elementary complexity.  In~\cite{bmmps16}, $\EXPSPACE$-hardness of the fragment $\BE$, and thus of full HS, has been shown.

The fragments $\AAbarBBbarEbar$ and $\AAbarEBbarEbar$ have been systematically studied in \cite{MMP15}. For each of them, an $\EXPSPACE$ model checking 
algorithm has been devised that, for any track of the Kripke structure, finds a satisfiability-preserving track of bounded length (\emph{track representative}). In this way, the model checking algorithm needs to check only tracks with a bounded maximum length. $\Psp$-hardness of the model checking problem for $\AAbarBBbarEbar$ and $\AAbarEBbarEbar$ has been proved in \cite{MMP15B} (if a succinct encoding of formulas is exploited, the algorithm remains in $\EXPSPACE$, but a $\NEXPTIME$ lower bound can be given~\cite{MMP15}).
Finally, it has been shown that formulas satisfying a constant bound on the nesting depth of $\hsB$ (respectively, $\hsE$) can be checked in polynomial working space~\cite{MMP15}.

Some well-behaved HS fragments, namely, $\AAbarBbarEbar$, $\Bbar$, $\Ebar$, $\AAbar$, $\A$, and $\Abar$, which are still expressive enough to capture meaningful interval properties of state transition systems and whose model checking problem has a computational complexity markedly lower than that of full HS, have been identified in \cite{MMP15B,MMPS16}. In particular the authors proved that the problem is $\Psp$-complete for the fragments $\AAbarBbarEbar$, $\Bbar$, and $\Ebar$, and in between $\Th$ and $\Thsq$~\cite{gottlob1995,schnoebelen2003} for $\AAbar$, $\A$, and $\Abar$. 
%
Two other well-behaved fragments, namely, $\AAbarBBbar$ and $\AAbarEEbar$, have been investigated in~\cite{bmmps16}, showing that their model checking problem is $\Psp$-complete. In addition, the authors showed that $\B$ and $\E$ are $\co\NP$-complete (the same complexity as the model checking problem for the purely propositional HS fragment $\HSprop$~\cite{MMP15B}).

In this paper, we complete the analysis of the sub-fragments of $\AAbarBBbar$ (resp., $\AAbarEEbar$). In Section~\ref{sec:AAbarBalgo}, we devise a $\PTIME^{\NP}$ model checking algorithm for $\AAbarB$ (resp., $\AAbarE$). Then, in Section~\ref{sec:ABhard}, we prove that $\AB$ (resp., $\AbarE$) is hard for $\PTIME^{\NP}$. It immediately follows that model checking for $\AB$ and $\AAbarB$ (resp., $\AbarE$ and $\AAbarE$) formulas over finite Kripke structures is $\PTIME^{\NP}$-complete. Finally, we show that $\AbarB$ (resp., $\AE$) are in $\Thsq$ (the proof is reported in~\cite{techrep}) and hard for $\Th$ (the hardness follows from that of $\Abar$, resp., $\A$~\cite{MMPS16}).

It is worth pointing out that the fragment $\AbarB$ belongs to a lower complexity class than the fragment $\AB$ (the same for the symmetric fragments $\AE$ and $\AbarE$). Such a difference can be explained as follows.

Let us consider a formula $\hsB \hsA \theta$, which belongs to $\A\B$. A track $\rho$ satisfies $\hsB \hsA \theta$ if there exists a prefix $\tilde{\rho}$ of $\rho$ from which a branch satisfying $\theta$ departs, i.e., a track starting from $\lst(\tilde{\rho})$. This amounts to say that $\A\B$ allows one to impose specific constraints on the branches departing from a state occurring in a given path. Such an ability will be exploited in Section~\ref{sec:ABhard} to prove the $\PTIME^{\NP}$-hardness of $\AB$.

Conversely, the fragment $\AbarB$ cannot express constraints of this form. For any given track $\rho$, modality $\hsAt$ only allows one to constrain tracks leading to the first state of $\rho$. As for modality $\hsB$, if we consider a prefix $\tilde{\rho}$ of $\rho$, the set of tracks leading to its first state is exactly the same as the set of those leading to the first state of $\rho$, as $\fst(\tilde{\rho}) = \fst(\rho)$. Therefore, pairing $\hsAt$ and $\hsB$ does not give any advantage in terms of expressiveness. Such a weakness of $\AbarB$ represents the reason why $\AbarB$ formulas can be checked in time $\Thsq$, instead of time $\PTIME^{\NP}$.

\section{A $\PTIME^{\NP}$ model checking algorithm for $\AAbarB$ formulas}\label{sec:AAbarBalgo}
In this section, we present a model checking algorithm for $\AAbarB$ formulas (Algorithm~\ref{MC}) belonging to the complexity class 
$\PTIME^{\NP}$. We recall that $\PTIME^{\NP}$ is the class of problems solvable in (deterministic) polynomial time exploiting an oracle for an $\NP$-complete problem. W.l.o.g., we restrict our attention to $\AAbarB$ formulas devoid of occurrences of conjunctions and universal modalities 
(definable, as usual, in terms of disjunctions, negations, and existential modalities).

\begin{algorithm}[tb]
\begin{algorithmic}[1]
	\For{all $\hsA\phi\in \mods(\psi)$}
		\State{\texttt{MC}$(\Ku,\phi,\textsc{forward})$}
	\EndFor
	\For{all $\hsAt\phi\in \mods(\psi)$}
		\State{\texttt{MC}$(\Ku,\phi,\textsc{backward})$}
	\EndFor
	
	\For{all $v\in W$}
		\If{\textsc{direction} is \textsc{forward}}
			\State{$V_{\A}(\psi,v)\gets Success(\texttt{Oracle}(\Ku,\psi,v,\textsc{forward},V_{\A}\cup V_{\Abar}))$}
		\ElsIf{\textsc{direction} is \textsc{backward}}
			\State{$V_{\Abar}(\psi,v)\gets Success(\texttt{Oracle}(\Ku,\psi,v,\textsc{backward},V_{\A}\cup V_{\Abar}))$}
		\EndIf
	\EndFor
\end{algorithmic}
\caption{\texttt{MC}$(\Ku,\psi,\textsc{direction})$}\label{MC}
\end{algorithm}

Algorithm~\ref{MC} presents the model checking procedure for a formula $\psi$ against a Kripke structure $\Ku$. It exploits two global vectors, $V_{\A}$ and $V_{\Abar}$, which can be seen as the tabular representations of two Boolean functions taking as arguments a subformula $\phi$ of $\psi$ and a state $v$ of $\Ku$. The intuition is that the function $V_{\A} (\phi,v)$ (resp., $V_{\Abar}(\phi,v)$) returns $\top$ if and only if there exists a 
track $\rho \in \Trk_\Ku$ starting from the state $v$ (resp., leading to the state $v$) such that $\Ku, \rho \models \phi$. 
The procedure $\texttt{MC}$ is initially invoked with parameters $(\Ku,\neg\psi,\textsc{forward})$. During the execution, it instantiates the entries of
$V_{\A}$ and $V_{\Abar}$, which are exploited in order to answer the model checking problem $\Ku\models \psi$; this is, in the end, equivalent to checking whether $V_{\A}(\neg\psi,w_0)=\bot$, where $w_0$ is the initial state of $\Ku$.

Let us consider the model checking procedure $\texttt{MC}$ in more detail. Besides the Kripke structure $\Ku$ and the formula $\psi$, $\texttt{MC}$ features a third parameter, $\textsc{direction}$, which can be assigned the value $\textsc{forward}$ (resp., $\textsc{backward}$), that is used in combination with the modality $\hsA$ (resp., $\hsAt$) for a forward (resp., backward) unravelling of $\Ku$. 
$\texttt{MC}$ is applied recursively on the nesting of modalities $\hsA$ and $\hsAt$ in the formula $\psi$ (in the base case, $\psi$ features no occurrences of $\hsA$ or $\hsAt$). In order to instantiate the Boolean vectors $V_{\A}$ and $V_{\Abar}$, an oracle is invoked (lines 5--9) for each state $v$ of the Kripke structure. Such an invocation is syntactically represented by $Success(\texttt{Oracle}(\Ku,\psi,v,\textsc{direction},V_{\A}\cup V_{\Abar}))$, and it returns $\top$ whenever there exists a computation of the non-deterministic algorithm $\texttt{Oracle}(\Ku,\psi,v,\textsc{direction},\allowbreak V_{\A}\cup V_{\Abar})$ returning $\top$, namely, whenever there is a suitable track starting from, or leading to $v$ (depending on the value of the parameter $\textsc{direction}$), and satisfying $\psi$.

We now introduce the notion of $\AAbar$-\emph{modal subformulas} of $\psi$; these subformulas ``direct'' the recursive calls of $\texttt{MC}$.
\begin{definition}
The set of $\AAbar$-\emph{modal subformulas} of an $\AAbarB$ formula $\psi$, denoted by $\mods(\psi)$, is the set of subformulas of $\psi$ having either the form $\hsA \psi'$ or the form $\hsAt \psi'$, for some $\psi'$, which are \emph{not in the scope of any $\hsA$ or $\hsAt$ modality}.
\end{definition}

For instance, $\mods(\hsA\hsAt q)=\{\hsA\hsAt q\}$ and $\mods\big(\big(\hsA p\, \wedge\, \hsA\hsAt q \big)\rightarrow \hsA p\big)\allowbreak =\{\hsA p,\hsA\hsAt q\}$.
%

$\texttt{MC}$ is recursively called on each formula $\phi$ such that $\hsA\phi$ or $\hsAt\phi$ belongs to $\mods(\psi)$ (lines 1--4). 
In this way, we can recursively gather in the Boolean vectors $V_{\A}$ and $V_{\Abar}$, by increasing nesting depth of the modalities $\hsA$ and $\hsAt$, the oracle answers for all the formulas $\psi'$ such that $\hsA\psi'$ or $\hsAt\psi'$ is a subformula (be it maximal or not) of $\psi$.

\begin{algorithm}[tb]
\resizebox{0.99\textwidth}{!}{
\begin{minipage}{1.1\linewidth}
\begin{multicols}{2}
\begin{algorithmic}[1]
\State{$\tilde{\rho}\gets \texttt{A\_track}(\Ku,v,|W|\cdot(2|\psi|+1)^2,\textsc{direction})$}\Comment{a track of $\Ku$ from/to $v$ of length $\leq |W|\cdot (2|\psi|+1)^2$}
\For{all $\hsA \phi\in\mods(\psi)$}
    	\For{$i=1,\cdots ,|\tilde{\rho}|$}
    		\State{$T[\hsA\phi,i]\gets V_{\A}(\phi,\tilde{\rho}(i))$}
    	\EndFor
\EndFor
\For{all $\hsAt \phi\in\mods(\psi)$}
    	\For{$i=1,\cdots ,|\tilde{\rho}|$}
    		\State{$T[\hsAt\phi,i]\gets V_{\Abar}(\phi,\fst(\tilde{\rho}))$}
    	\EndFor
\EndFor

\For{all subformulas $\varphi$ of $\psi$, not contained in (or equal to) $\AAbar$-modal subformulas of $\psi$, by increasing length}
    \If{$\varphi=p$, for $p\in\mathpzc{AP}$}
        \State{$T[p,1]\gets p\in \mu(\fst(\tilde{\rho}))$}
        \For{$i=2,\cdots ,|\tilde{\rho}|$}
            \State{$T[p,i]\gets T[p,i-1]$ and $p\in \mu(\tilde{\rho}(i))$}
        \EndFor
    		
\columnbreak 

	\ElsIf{$\varphi=\neg \varphi_1$}
	    \For{$i=1,\cdots ,|\tilde{\rho}|$}
            \State{$T[\varphi,i]\gets$ not $T[\varphi_1,i]$}
        \EndFor        
    \ElsIf{$\varphi=\varphi_1\vee\varphi_2$}
        \For{$i=1,\cdots ,|\tilde{\rho}|$}
            \State{$T[\varphi,i]\gets T[\varphi_1,i]$ or $T[\varphi_2,i]$}
        \EndFor
	\ElsIf{$\varphi=\hsB\varphi_1$}
	    \State{$T[\varphi,1]\gets\bot$}
	    \For{$i=2,\cdots ,|\tilde{\rho}|$}
            \State{$T[\varphi,i]\gets T[\varphi,i-1]$ or $T[\varphi_1,i-1]$}
	    \EndFor
	\EndIf
\EndFor	
\Return{$T[\psi,|\tilde{\rho}|]$}
\end{algorithmic}
\end{multicols}
\end{minipage}}
\caption{\texttt{Oracle}$(\Ku,\psi,v,\textsc{direction},V_{\A}\cup V_{\Abar})$}\label{Oracle}
\end{algorithm}

Let us now consider the \emph{non-deterministic polynomial time} procedure $\texttt{Oracle}(\Ku,\psi,v,\textsc{direction},\allowbreak V_{\A}\cup V_{\Abar})$ reported in Algorithm~\ref{Oracle}, which is used as the basic engine by the oracle in the aforementioned model checking Algorithm~\ref{MC}. 
The idea underlying Algorithm~\ref{Oracle} is first to non-deterministically generate a track $\tilde{\rho}$ by unravelling the Kripke structure $\Ku$
according to the parameter $\textsc{direction}$, and then to verify $\psi$ over $\tilde{\rho}$.
Such a procedure actually exploits a result proved in \cite{bmmps16} (see, in particular, Theorem~10) stating a so-called ``polynomial-size model-track property'' for formulas of the fragment $\AAbarEEbar$: if $\rho$ is a track of $\mathpzc{K}$, $\phi$ is an $\AAbarEEbar$ formula, and $\Ku,\rho \models \phi$, then there exists  $\rho' \in \Trk_\mathpzc{K}$ such that $|\rho'|\leq |W|\cdot (2|\phi|+1)^2$, $\fst(\rho)=\fst(\rho')$, $\lst(\rho)=\lst(\rho')$, and $\Ku,\rho' \models \phi$. This property guarantees that, in order to check the satisfiability of a formula $\phi$, it is enough to consider tracks having a length bounded by $|W|\cdot (2|\phi|+1)^2$. Such a result holds by symmetry for formulas of the fragment $\AAbarBBbar$ as well.

An execution of \texttt{Oracle}$(\Ku,\psi,v,\textsc{direction},V_{\A}\cup V_{\Abar})$ starts (line 1) by \emph{non-deterministically} generating a track $\tilde{\rho}$ (having a length of at most $|W|\cdot (2|\psi|+1)^2$), with $v$ as its first (resp., last) state if the \textsc{direction} parameter is \textsc{forward} (resp., \textsc{backward}). The track is generated by visiting the unravelling of $\Ku$ (resp., of $\Ku$ with transposed edges). The remaining part of the algorithm checks \emph{deterministically} whether $\Ku,\tilde{\rho}\models\psi$. Such a verification is performed in a bottom-up way: for all the subformulas $\phi$ of $\psi$ (starting from the minimal ones) and for all the prefixes $\tilde{\rho}(1,i)$ of $\tilde{\rho}$, with $1 \leq i \leq  |\tilde{\rho}|$ (starting from the shorter ones), the procedure establishes whether $\Ku,\tilde{\rho}(1,i)\models\phi$ or not, and this result is stored in the entry $T[\phi,i]$ of a Boolean table $T$. Note that if the considered subformula of $\psi$ is an element of $\mods(\psi)$, the algorithm does not need to perform any verification, since the result is already available in the Boolean vectors $V_{\A}$ and $V_{\Abar}$ (as a consequence of the previously completed calls to the procedure \texttt{Oracle}), and the table $T$ is updated accordingly (lines 2--7).
For the remaining subformulas, the entries of $T$ are computed, as we already said, in a bottom-up fashion (lines 8--22). The result of the overall verification is stored 
in $T[\psi,|\tilde{\rho}|]$ and returned (line 23).

The algorithm presented here for checking formulas of $\AAbarB$ can trivially be adapted to check formulas of the symmetric fragment $\AAbarE$.

The following lemma establishes the soundness and completeness of the procedure \texttt{Oracle}.

\begin{lemma}\label{lemmaOracle}
Let $\Ku=(\mathpzc{AP},W, \delta,\mu,w_0)$ be a finite Kripke structure, $\psi$ be an $\AAbarB$ formula, and $V_{\A}(\bullet,\bullet)$ and $V_{\Abar}(\bullet,\bullet)$ be two Boolean arrays. Let us assume that 
\begin{compactenum}
	\item for each $\hsA \phi\in\mods(\psi)$ and $v'\in W$, $V_{\A}(\phi,v')=\top$ iff there exists $\rho\in\Trk_\Ku$ such that $\fst(\rho)=v'$ and $\Ku,\rho\models \phi$, and
	\item for each $\hsAt \phi\in\mods(\psi)$ and $v'\in W$, $V_{\Abar}(\phi,v')=\top$ iff there exists $\rho\in\Trk_\Ku$ such that $\lst(\rho)=v'$ and $\Ku,\rho\models \phi$.
\end{compactenum}
Then, \texttt{Oracle}$(\Ku,\psi,v,\textsc{direction},V_{\A}\cup V_{\Abar})$ features a successful computation (returning $\top$) iff:
\begin{compactitem}
	\item there exists $\rho\in\Trk_\Ku$ such that $\fst(\rho)=v$ and $\Ku,\rho\models \psi$, in the case \textsc{direction} is \textsc{forward};
	\item there exists $\rho\in\Trk_\Ku$ such that $\lst(\rho)=v$ and $\Ku,\rho\models \psi$, in the case \textsc{direction} is \textsc{backward}.
\end{compactitem}
\end{lemma}
\begin{proof}
 It is easy to check that if $\tilde{\rho}$ is the track non-deterministically generated by \texttt{A\_track} at line 1,  then, for $i=1,\cdots, |\tilde{\rho}|$, it holds that $\Ku,\tilde{\rho}(1,i)\models\phi\iff T[\phi,i]=\top$, either by hypothesis, when  $\phi$ occurs in  $\mods(\psi)$ (lines 2--7), or by construction, when $\phi$ does not occur in $\mods(\psi)$ (lines 8--22).

Let us now assume that the value of the parameter \textsc{direction} is \textsc{forward} (the proof for the other direction is analogous).

\begin{compactitem}
	\item[$(\Rightarrow)$] If \texttt{Oracle}$(\Ku,\psi,v,\textsc{forward},V_{\A}\cup V_{\Abar})$ features a successful computation, it means that there exists a track $\tilde{\rho} \in\Trk_\Ku$ (generated at line 1) such that $\fst(\tilde{\rho})=v$ and $T[\psi,|\tilde{\rho}|]=\top$. Hence $\Ku,\tilde{\rho}\models\psi$.

	\item[$(\Leftarrow)$] If there exists $\rho\in\Trk_\Ku$ such that $\fst(\rho)=v$ and $\Ku,\rho\models \psi$, as a result of Theorem~10 of \cite{bmmps16}, there exists $\tilde{\rho}\in\Trk_\Ku$ such that $\Ku,\tilde{\rho}\models \psi$, $\fst(\tilde{\rho})=\fst(\rho)$, and $|\tilde{\rho}|\leq
	|W|\cdot (2|\psi|+1)^2$.
	It follows that in some non-deterministic instance of \texttt{Oracle}$(\Ku,\psi,v,\textsc{forward},V_{\A}\cup V_{\Abar})$, $\texttt{A\_track}(\Ku,v,|W|\cdot(2|\psi|+1)^2,\textsc{forward})$ returns such $\tilde{\rho}$ (at line 1). Finally, we have that $T[\psi,|\tilde{\rho}|]=\top$ as $\Ku,\tilde{\rho}\models \psi$, hence the considered instance of \texttt{Oracle}$(\Ku,\psi,v,\textsc{forward},V_{\A}\cup V_{\Abar})$ is successful.\qedhere
\end{compactitem}
\end{proof}

The following theorem states soundness and completeness of the model checking procedure \texttt{MC}.

\begin{theorem}\label{MCsoundCompl}
Let $\Ku=(\mathpzc{AP},W, \delta,\mu,w_0)$ be a finite Kripke structure, $\psi$ be an $\AAbarB$ formula, and $V_{\A}(\bullet,\bullet)$ and $V_{\Abar}(\bullet,\bullet)$ be two Boolean arrays. If $\texttt{MC}(\Ku,\psi,\textsc{direction})$ is executed, then for all $v\in W$:
\begin{compactitem}
	\item if \textsc{direction} is \textsc{forward}, $V_{\A}(\psi,v)=\top$ iff there is $\rho\in\Trk_\Ku$ such that $\fst(\rho)=v$ and $\Ku,\rho\models \psi$;
	\item if \textsc{direction} is \textsc{backward}, $V_{\Abar}(\psi,v)=\top$ iff there is $\rho\in\Trk_\Ku$ such that $\lst(\rho)=v$ and $\Ku,\rho\!\models\!\psi$.
\end{compactitem}
\end{theorem}

\begin{proof}
The proof is by induction on the number $n$ of occurrences of $\hsA$ and $\hsAt$ modalities in $\psi$.
\newline
(Base case: $n=0$) Since $\mods(\psi)=\emptyset$, conditions 1 and 2 of Lemma \ref{lemmaOracle} are satisfied and the thesis trivially holds.
\newline
(Inductive case: $n>0$) The formula $\psi$ contains at least an $\hsA$ or an $\hsAt$ modality. Hence $\mods(\psi)\allowbreak \neq \emptyset$. Since each recursive call to \texttt{MC} (either at line 2 or 4) is performed on a formula $\phi$ featuring a number of occurrences of $\hsA$ and $\hsAt$ which is strictly less than the number of their occurrences in $\psi$, we can apply the inductive hypothesis. As a consequence, when the control flow reaches line 5, it holds that:
\begin{compactenum}
	\item for each $\hsA \phi\in\mods(\psi)$ and $v'\in W$, $V_{\A}(\phi,v')=\top$ iff there exists $\rho\in\Trk_\Ku$ such that $\fst(\rho)=v'$ and $\Ku,\rho\models \phi$;
	\item for each $\hsAt \phi\in\mods(\psi)$ and $v'\in W$, $V_{\Abar}(\phi,v')=\top$ iff there exists $\rho\in\Trk_\Ku$ such that $\lst(\rho)=v'$ and $\Ku,\rho\models \phi$.
\end{compactenum}
This implies that conditions 1 and 2 of Lemma~\ref{lemmaOracle} are fulfilled. Hence (assuming that \textsc{direction} is \textsc{forward}), it holds that, for $v\in W$, $V_{\A}(\psi,v)=\top$ iff there exists $\rho\in\Trk_\Ku$ such that $\fst(\rho)=v$ and $\Ku,\rho\models \psi$. The case for  \textsc{direction} $=$  \textsc{backward} is symmetric, and thus omitted.
\end{proof}

As an immediate consequence we have that the procedure  \texttt{MC} solves the model checking problem for $\AAbarB$ with an  algorithm belonging to the complexity class $\PTIME^{\NP}$.

\begin{corollary}
Let $\Ku=(\mathpzc{AP},W, \delta,\mu,w_0)$ be a finite Kripke structure and $\psi$ be an $\AAbarB$ formula. If $\texttt{MC}(\Ku,\neg\psi,\textsc{forward})$ is executed, then $V_{\A}(\neg\psi,w_0)=\bot \iff \Ku\models \psi$.
\end{corollary}

\begin{corollary}
The model checking problem for $\AAbarB$ formulas over finite Kripke structures is in $\PTIME^{\NP}$.
\end{corollary}
\begin{proof}
Given a finite Kripke structure $\Ku=(\mathpzc{AP},W, \delta,\mu,w_0)$ and an $\AAbarB$ formula $\psi$, the number of recursive calls performed by $\texttt{MC}(\Ku,\neg\psi,\textsc{forward})$ is at most $|\psi|$. Each one costs $O(|\psi|+|W|\cdot (|\Ku|+|\psi|+|\psi|\cdot |W|))$, where the first addend comes from searching $\psi$ for its modal subformulas (lines 1--4), and the second one from the preparation of the input for the oracle call, for each $v\in W$ (lines 5--9). Therefore its (deterministic) complexity is $O(|\psi|^2\cdot |\Ku|^2)$. 
As for \texttt{Oracle}$(\Ku,\psi,v,\textsc{direction},V_{\A}\cup V_{\Abar})$, its (non-deterministic) complexity is $O(|\psi|^3\cdot |\Ku|)$, where $|\psi|$ is a bound to the number of subformulas and $O(|\psi|^2 \cdot |\Ku|)$ is the number of steps necessary to generate and check $\tilde{\rho}$.
\end{proof}

Symmetrically, by easily adapting the procedure \texttt{Oracle}, it is straightforward to prove that the model checking problem for $\AAbarE$ formulas is in $\PTIME^{\NP}$ as well.

\section{$\PTIME^{\NP}$-hardness of model checking for $\AB$ formulas}\label{sec:ABhard}

In this section, we prove that model checking for $\AB$ (and $\AbarE$) formulas is hard for $\PTIME^{\NP}$ by reducing the $\PTIME^{\NP}$-complete problem SNSAT (Sequentially Nested SATisfiability), a logical problem with nested satisfiability questions~\cite{LMS01}, to it. SNSAT is defined as follows.

\begin{definition}\label{snsat}
An instance $\mathcal{I}$ of SNSAT consists of a set of Boolean variables $X=\{x_1,\cdots ,x_n\}$ and a set of Boolean formulas $\{F_1(Z_1), F_2(x_1,Z_2),\cdots , F_n(x_1,\cdots , x_{n-1}, Z_n)\}$, where, for $i\!=\!1,\cdots , n$, $F_i(x_1,\! \cdots\! ,x_{i-1},Z_i)$ features variables in $\{x_1, \cdots ,x_{i-1}\}$ and in $Z_i=\{z_i^1,\cdots ,z_i^{j_i}\}$, the latter being a set of variables local to $F_i$, that is, $Z_i\cap Z_j=\emptyset$, for $j \neq i$, and $X\cap Z_i=\emptyset$. We denote $|X| (=n)$ by $|\mathcal{I}|$.
Let $v_\mathcal{I}$ be the valuation of the variables in $X$ defined as follows: $v_\mathcal{I}(x_i)=\top \iff F_i(v_\mathcal{I}(x_1), \cdots, v_\mathcal{I}(x_{i-1}), Z_i)$ is satisfiable (by assigning suitable values to the local variables $z_i^1,\cdots ,z_i^{j_i}\in Z_i$).
SNSAT is the problem of deciding, given an instance $\mathcal{I}$, with $|\mathcal{I}|=n$, whether
 $v_\mathcal{I}(x_n)=\top$. In such a case, we say that $\mathcal{I}$ is a positive instance of SNSAT.
\end{definition}

Given an SNSAT instance $\mathcal{I}$, with $|\mathcal{I}|=n$, the valuation $v_\mathcal{I}$ is unique and it can be easily computed by a $\PTIME^{\NP}$ algorithm as follows. A first query to a SAT oracle determines whether $v_\mathcal{I}(x_1)$ is $\top$ or $\bot$, since $v_\mathcal{I}(x_1)=\top$ iff $F_1(Z_1)$ is satisfiable. Then, we replace $x_1$ by the value $v_\mathcal{I}(x_1)$ in $F_2(x_1,Z_2)$ and another query to the SAT oracle is performed to determine whether $F_2(v_\mathcal{I}(x_1),Z_2)$ is satisfiable, gaining the value of $v_\mathcal{I}(x_2)$. This step is iterated other $n-2$ times, until the value for $v_\mathcal{I}(x_n)$ is obtained. 
%

Let $\mathcal{I}$ be an instance of SNSAT, with $|\mathcal{I}|=n$. We now show how to build a finite Kripke structure $\Ku_\mathcal{I}$ and an $\AB$ formula $\Phi_\mathcal{I}$, by using \emph{logarithmic working space}, such that $\mathcal{I}$ is a positive instance of SNSAT if and only if $\Ku_\mathcal{I}\models \Phi_\mathcal{I}$. Such a reduction is inspired by similar constructions from \cite{LMS01}.

Let $Z=\bigcup_{i=1}^n Z_i$ and let $R=\{r_i \mid i=1,\cdots , n\}$ and $R_i=R\setminus \{r_i\}$ be $n+1$ sets of auxiliary variables.
The Kripke structure $\Ku_\mathcal{I}$ consists of a suitable composition of $n$ instances of a \emph{gadget} (an instance  for each variable $x_1,\cdots , x_n\in X$). The structure of the gadget for $x_i$, with $1\leq i\leq n$, is shown in Figure~\ref{gadget}, assuming that the labeling of states (nodes) is defined as follows:
\begin{compactitem}
	\item $\mu(w_{x_i})=X \cup Z \cup \{s,t\} \cup R_i$, and
    	$\mu(\overline{w_{x_i}})=(X\setminus \{x_i\}) \cup Z \cup \{s,t\} \cup R_i \cup \{p_{\overline{x_i}}\}$;
	\item for $u_i=1,\cdots , j_i$, $\mu(w_{z_i^{u_i}})=X \cup Z \cup \{s,t\} \cup R_i$, and
	    $\mu(\overline{w_{z_i^{u_i}}})=X \cup (Z\setminus \{z_i^{u_i}\}) \cup \{s,t\} \cup R_i$;
	\item $\mu(\overline{s_i})=X \cup Z \cup \{t\} \cup R_i$.
\end{compactitem}

\begin{figure}[t]
 \begin{minipage}[b]{0.32\textwidth}
   \centering
   \scalebox{0.8}{\begin{tikzpicture}[every node/.style={circle, draw, inner sep=1pt}]

\node (v3) at (-2,0.5) {$w_{x_i}$};
\node (v2) at (0,0.5) {$\overline{s_i}$};
\node (v1) at (2,0.5) {$\overline{w_{x_i}}$};
\draw [->] (v1) edge (v2);
\draw [->] (v2) edge (v3);
\draw [dashed] (-2.5,6) rectangle (2.5,1.5);

\draw [use as bounding box,draw=none] (-2.5,6) rectangle (2.5,0.2);

\node (v4) at (-2,2) {$w_{z_i^1}$};
\node (v6) at (2,2) {$\overline{w_{z_i^1}}$};
\node (v5) at (-2,3.5) {$w_{z_i^2}$};
\node (v7) at (2,3.5) {$\overline{w_{z_i^2}}$};
\node (v8) at (-2,5.5) {$w_{z_i^{j_i}}$};
\node (v9) at (2,5.5) {$\overline{w_{z_i^{j_i}}}$};
\draw [->] (v4) edge (v5);
\draw [->] (v6) edge (v7);
\draw [->] (v4) edge (v7);
\draw [->] (v6) edge (v5);

\draw [->,dotted] (v5) edge (v8);
\draw [->,dotted] (v7) edge (v9);
\draw [->,dotted] (v7) edge (v8);
\draw [->,dotted] (v5) edge (v9);
\draw [->] (v3) edge (v4);
\draw [->] (v3) edge (v6);
\draw [->] (v1) edge (v4);
\draw [->] (v1) edge (v6);

\node [draw=none,rotate=-90] (ch) at (2.8,3.5) {choice $Z_i$};
\end{tikzpicture}}
   \caption{The gadget for $x_i$.}\label{gadget}
 \end{minipage}
 \quad
 \begin{minipage}[b]{0.64\textwidth}
   \centering
   \rotatebox{-90}{\scalebox{0.8}{\begin{tikzpicture}[every node/.style={circle, draw, inner sep=1pt}]

\node [rotate=90] (v3) at (-2,0.5) {$w_{x_1}$};
\node [rotate=90] (v2) at (0,0.5) {$\overline{s_1}$};
\node [rotate=90] (v1) at (2,0.5) {$\overline{w_{x_1}}$};
\draw [->] (v1) edge (v2);
\draw [->] (v2) edge (v3);
\draw  (-2.5,3) rectangle (2.5,1.5) node[midway,draw=none,rotate=30] {choice $Z_1$};;

\node [draw=none] (v4) at (-2,1.5) {};
\node [draw=none] (v6) at (2,1.5) {};

\draw [->] (v3) edge (v4);
\draw [->] (v1) edge (v6);

\node [double,rotate=90]  (v13) at (-2,-8) {$w_{x_n}$};
\node [rotate=90] (v12) at (0,-8) {$\overline{s_n}$};
\node [rotate=90] (v11) at (2,-8) {$\overline{w_{x_n}}$};
\draw [->] (v11) edge (v12);
\draw [->] (v12) edge (v13);
\draw  (-2.5,-5.5) rectangle (2.5,-7) node[midway,draw=none,rotate=30] {choice $Z_n$};;

\node [draw=none] (v14) at (-2,-7) {};
\node [draw=none] (v16) at (2,-7) {};

\draw [->] (v13) edge (v14);
\draw [->] (v11) edge (v16);

\node [rotate=90] (v23) at (-2,-3.5) {$w_{x_2}$};
\node [rotate=90] (v22) at (0,-3.5) {$\overline{s_2}$};
\node [rotate=90] (v21) at (2,-3.5) {$\overline{w_{x_2}}$};
\draw [->] (v21) edge (v22);
\draw [->] (v22) edge (v23);
\draw  (-2.5,-1) rectangle (2.5,-2.5) node[midway,draw=none,rotate=30] {choice $Z_2$};;

\node [draw=none] (v24) at (-2,-2.5) {};
\node [draw=none] (v26) at (2,-2.5) {};

\draw [->] (v23) edge (v24);
\draw [->] (v21) edge (v26);

\node [draw=none] (v8) at (-2,-5.5) {};
\node [draw=none] (v18) at (2,-5.5) {};
\node [draw=none] (v29) at (2,-1) {};
\node [draw=none] (v19) at (-2,-1) {};
\draw [->, dotted] (v8) edge (v23);
\draw [->, dotted] (v18) edge (v21);
\draw [->] (v19) edge (v3);
\draw [->] (v29) edge (v1);

\node [rotate=90] (v9) at (0,3.5) {$s_0$};
\node [draw=none] (v30) at (-2,3) {};
\draw [->,looseness=5] (v9.south east) edge (v9.north east);

\node [draw=none] (v10) at (2,3) {};

\draw [->] (v19) edge (v1);
\draw [->] (v29) edge (v3);
\draw [->,dotted] (v18) edge (v23);
\draw [->,dotted] (v8) edge (v21);
\draw [->] (v30) edge (v9);
\draw [->] (v10) edge (v9);
\draw [->] (v13) edge (v16);
\draw [->] (v11) edge (v14);
\draw [->] (v23) edge (v26);
\draw [->] (v21) edge (v24);
\draw [->] (v3) edge (v6);
\draw [->] (v1) edge (v4);

\draw [dashed] (-2.7,3.2) rectangle (2.7,0);
\draw [dashed] (-2.7,-0.8) rectangle (2.7,-4);
\draw [dashed] (-2.7,-5.3) rectangle (2.7,-8.5);

\end{tikzpicture}}}
   \caption{Kripke structure $\Ku_\mathcal{I}$ associated with an SNSAT instance $\mathcal{I}$, with $|\mathcal{I}|=n$. Notice that the states $\overline{s_n}$ and $\overline{w_{x_n}}$ are unreachable.}\label{fullKripke}
 \end{minipage}
\end{figure}

The Kripke structure  $\Ku_\mathcal{I}$ is obtained by sequentializing (adding suitable arcs) the $n$ instances of the gadget (in reverse order, from $x_n$ to $x_1$), adding a collector terminal state $s_0$, with labeling $\mu(s_0)=X \cup Z \cup \{s\} \cup R$, and setting $w_{x_n}$ as the initial state. The overall construction is reported in Figure~\ref{fullKripke}. Formally, $\Ku_\mathcal{I}=(X \cup Z \cup \{s,t\} \cup R\cup\{p_{\overline{x_i}}\mid i=1,\cdots ,n\}, W,\delta,\mu,w_{x_n})$. $\Ku_\mathcal{I}$ enjoys the following properties:
	$(i)$ any track satisfying $s$ does not pass through any $\overline{s_i}$, for $1\leq i\leq n$;
	$(ii)$ any track \emph{not} satisfying $t$ has $s_0$ as its last state;
	$(iii)$ any track \emph{not} satisfying $r_i$ passes through some state of the $i$-th gadget, for $1\leq i\leq n$;
	$(iv)$ the only track satisfying $p_{\overline{x_i}}$ is $\overline{w_{x_i}}$ (notice that $|\overline{w_{x_i}}|=1$), for $1\leq i\leq n$.

A track $\rho\in\Trk_{\Ku_\mathcal{I}}$ \emph{induces} a truth assignment of all the proposition letters, denoted by $\omega_\rho$, which is defined as $\omega_\rho(y)=\top \iff \Ku_\mathcal{I},\rho\models y$, for any letter $y$. 
In the following, we will write $\omega_\rho(Z_i)$ for $\omega_\rho(z_i^1),\cdots , \omega_\rho(z_i^{j_i})$.
In particular, if $\rho$ starts from some state $w_{x_i}$ or $\overline{w_{x_i}}$, and satisfies $s\wedge \neg t$ (that is, it reaches the collector state $s_0$ without visiting any node $\overline{s_j}$, for $1\leq j \leq i$), $\omega_\rho$ fulfills the following conditions: 
for $1\leq m\leq i$,
\begin{compactitem}
	\item if $w_{x_m}\in\states(\rho)$, then $\omega_\rho(x_m)=\top$, and if $\overline{w_{x_m}}\in\states(\rho)$, then $\omega_\rho(x_m)=\bot$;
	\item for $1\leq u_m\leq j_m$, if $w_{z_m^{u_m}}\in\states(\rho)$, then $\omega_\rho(z_m^{u_m})=\top$, and if $\overline{w_{z_m^{u_m}}}\in\states(\rho)$, then $\omega_\rho(z_m^{u_m})=\bot$;
\end{compactitem}	
It immediately follows that $\Ku_\mathcal{I},\rho\models F_m(x_1,\cdots , x_{m-1},Z_m)$ iff 
$F_m(\omega_\rho(x_1),\cdots , \omega_\rho(x_{m-1}),\omega_\rho(Z_m)) = \top$.
Finally, let $\mathcal{F}_\mathcal{I}=\{\psi_k \mid 0\leq k\leq n+1\}$ be the set of formulas defined as: $\psi_0=\bot$ and, for $k\geq 1$,
\begin{equation*}
\psi_k = \hsA \underbrace{\left[\begin{array}{c}
(s\wedge \neg t) \wedge \bigwedge_{i=1}^n \Big((x_i\wedge \neg r_i)\rightarrow F_i(x_1, \cdots, x_{i-1}, Z_i)\Big) \\ 
\wedge \\ 
\mathopen[B\mathclose] \Big( (\bigvee_{i=1}^n \hsA p_{\overline{x_i}})\rightarrow \hsA\big(\neg s \wedge \ell_{=2}\wedge \hsA (\ell_{=2}\wedge \neg\psi_{k-1})\big)\Big)
\end{array}\right]}_{\text{\normalsize $\varphi_k$}} ,
\end{equation*}
where $\ell_{=2}=\hsB \top \wedge [B][B]\bot$ is satisfied only by tracks of length 2.
The first conjunct of $\varphi_k$ ($s\wedge \neg t$) forces the track to reach the collector state $s_0$, without visiting any state $\overline{s_j}$. The second conjunct checks that if the track assigns the truth value $\top$ to $x_m$ passing through $w_{x_m}$ (with $1\leq m \leq n$), then $F_m(x_1,\cdots,x_{m-1},Z_m)$ is satisfied by $\omega_\rho$ (which amounts to say that the SAT problem connected with $Z_m$ has a positive answer, for the selected values of $x_1,\cdots,x_{m-1}$).
Conversely, the third conjunct ensures that if the track assigns the truth value $\bot$ to some $x_m$ by passing through $\overline{w_{x_m}}$, then, intuitively, the SAT problem connected with $Z_m$ has no assignment satisfying $F_m(x_1,\cdots,x_{m-1},Z_m)$.
As a matter of fact,
if $\rho$ satisfies $\varphi_k$ for some $k\geq 2$, and assigns $\bot$ to $x_m$, then there is a prefix $\tilde{\rho}$ of $\rho$ ending in $\overline{w_{x_m}}$. Since  $\bigvee_{i=1}^n\hsA p_{\overline{x_i}}$ is satisfied by $\tilde{\rho}$, then  $\hsA\big(\neg s \wedge \ell_{=2}\wedge \hsA (\ell_{=2}\wedge \neg\psi_{k-1})\big)$ must be satisfied as well. The only possibility is that the track $\overline{s_m}\cdot w_{x_m}$ does not model $\psi_{k-1}$ (as $\overline{w_{x_m}}\cdot \overline{s_m}$ has to model $\hsA (\ell_{=2}\wedge \neg\psi_{k-1})$). However, since $\psi_{k-1}=\hsA\varphi_{k-1}$, this holds iff $\Ku, w_{x_m}\not \models\psi_{k-1}$.

The following theorem states the correctness of the construction.

\begin{theorem}\label{thcorr} Let $\mathcal{I}$ be an instance of SNSAT with $|\mathcal{I}|=n$, and let $\Ku_\mathcal{I}$ and $\mathcal{F}_\mathcal{I}$
be defined as above. For all $0\leq k\leq n+1$ and all $r=1,\cdots , n$, it holds that:
	\begin{compactenum}
		\item if $k\geq r$, then $v_\mathcal{I}(x_r)=\top \iff \Ku_\mathcal{I},w_{x_r}\models \psi_k$;
		\item if $k\geq r+1$, then $v_\mathcal{I}(x_r)=\bot \iff \Ku_\mathcal{I},\overline{w_{x_r}}\models \psi_k$.
	\end{compactenum}
\end{theorem}
\begin{proof}
The proof is by induction on $k\geq 0$. 
\newline (Base case: $k=0$). The thesis trivially holds.
\newline
(Inductive case: $k\geq 1$). We first prove the $(\Leftarrow)$ implication for both item 1 and item 2.
\begin{compactitem}
\item (Item 1) Assume that $k\geq r$ and $\Ku_\mathcal{I},w_{x_r}\models \psi_k$. Thus, there exists $\rho\in\Trk_{\Ku_\mathcal{I}}$ such that $\rho=w_{x_r}\cdots s_0$ does not pass through any $\overline{s_m}$, $1\leq m\leq r$ and $\Ku_\mathcal{I},\rho\models \varphi_k$. We show by induction on $1\leq m\leq r$ that $\omega_\rho(x_m)=v_\mathcal{I}(x_m)$. 

	\begin{compactitem}
		\item Let us consider first the case where $\rho$ passes through $w_{x_m}$, implying that $\omega_\rho(x_m)=\top$; thus $\Ku_\mathcal{I},\rho\models x_m\wedge \neg r_m$ and $\Ku_\mathcal{I},\rho\models F_m(x_1,\cdots ,x_{m-1},Z_m)$. If $m=1$ (base case), since $F_1$ is satisfiable, then $v_\mathcal{I}(x_1)=\top$. If $m\geq 2$ (inductive case), by the inductive hypothesis, it holds that $\omega_\rho(x_1)=v_\mathcal{I}(x_1)$, \dots , $\omega_\rho(x_{m-1})=v_\mathcal{I}(x_{m-1})$. Since $\Ku_\mathcal{I},\rho\models F_m(x_1,\cdots ,x_{m-1},Z_m)$ or, equivalently, $F_m(\omega_\rho(x_{1}),\cdots , \omega_\rho(x_{m-1}), \omega_\rho(Z_m))=\top$, it holds that $F_m(v_\mathcal{I}(x_{1}),\!\cdots\! , v_\mathcal{I}(x_{m-1}), \allowbreak  \omega_\rho(Z_m))=\top$ and, by definition of $v_\mathcal{I}$, $v_\mathcal{I}(x_m)=\top$.
		
		\item Conversely, let us consider the case where $\rho$ passes through $\overline{w_{x_m}}$, implying that $\omega_\rho(x_m)=\bot$ and $m<r$, as we are assuming $\fst(\rho)=w_{x_r}$. In this case, the prefix $w_{x_r}\cdots \overline{w_{x_m}}$ of $\rho$ satisfies both $\bigvee_{i=1}^n \hsA p_{\overline{x_i}}$ and $\hsA\big(\neg s \wedge \ell_{=2}\wedge \hsA (\ell_{=2}\wedge \neg\psi_{k-1})\big)$. Therefore, $\Ku_\mathcal{I},\overline{w_{x_m}}\cdot \overline{s_m}\models \hsA (\ell_{=2}\wedge \neg\psi_{k-1})$ and $\Ku_\mathcal{I}, \overline{s_m}\cdot w_{x_m}\not\models \psi_{k-1}$, with $\psi_{k-1}=\hsA\varphi_{k-1}$. Hence $\Ku_\mathcal{I}, w_{x_m}\not\models \psi_{k-1}$. Since $1\leq m<r$, we have $1\leq m<r\leq k$, thus $k'=k-1\geq m\geq 1$. By the inductive hypothesis (on $k'=k-1$), we get that $v_\mathcal{I}(x_m)=\bot$.
	\end{compactitem}
Therefore $v_\mathcal{I}(x_r)=\omega_\rho(x_r)$ and, since $w_{x_r}\in\states(\rho)$, we have that $\omega_\rho(x_r)=\top$ and the thesis, that is, $v_\mathcal{I}(x_r)=\top$, follows.
	
\item (Item 2) Assume that $k\geq r+1$ and $\Ku_\mathcal{I},\overline{w_{x_r}}\models \psi_k$. The proof follows the same steps as the previous case and it is thus only sketched: there exists $\rho\in\Trk_{\Ku_\mathcal{I}}$ such that $\rho=\overline{w_{x_r}}\cdots s_0$ does not pass through any $\overline{s_m}$, for $1\leq m\leq r$, and $\Ku_\mathcal{I},\rho\models \varphi_k$. The only thing which changes is that the prefix $\overline{w_{x_r}}$ satisfies $\bigvee_{i=1}^n \hsA p_{\overline{x_i}}$, thus as before we get $\Ku_\mathcal{I}, w_{x_r}\not\models \psi_{k-1}$. Now, $k'=k-1\geq r\geq 1$ and, by the inductive hypothesis (on $k'=k-1$), it holds that $v_\mathcal{I}(x_r)=\bot$.

\end{compactitem}

We prove now the converse implication $(\Rightarrow)$ for both item 1 and item 2.
\begin{compactitem}
\item (Item 1) Assume that $k\geq r$ and $v_\mathcal{I}(x_r)=\top$. Let us consider the track $\rho\in\Trk_{\Ku_\mathcal{I}}$, $\rho=w_{x_r}\cdots s_0$ never passing through any $\overline{s_m}$, for $1\leq m\leq r$, such that $w_{x_m}\in\states(\rho)$ if $v_\mathcal{I}(x_m)=\top$, and $\overline{w_{x_m}}\in\states(\rho)$ if $v_\mathcal{I}(x_m)=\bot$, for $1\leq m\leq r$. Such a choice of 
$\rho$ ensures that $v_\mathcal{I}(x_m)=\omega_\rho(x_m)$. In addition, the choice of $\rho$ has to induce also the proper evaluation of local variables,
that is, if $v_\mathcal{I}(x_m)=\top$, then for $1\leq u_m\leq j_m$, $w_{z_m^{u_m}}\in\states(\rho)$ if $F_m(v_\mathcal{I}(x_1),\cdots , v_\mathcal{I}(x_{m-1}),Z_m)$ is satisfied for $z_m^{u_m}= \top$, $\overline{w_{z_m^{u_m}}}\in\states(\rho)$ otherwise. Notice that such a choice of $\rho$ is always possible.
We have to show that $\Ku_\mathcal{I},\rho\models \varphi_k$, hence $\Ku_\mathcal{I},w_{x_r}\models \psi_k$.
\begin{compactitem}
	\item For all $1\leq m\leq r$ such that $v_\mathcal{I}(x_m)=\top$, it holds that $F_m(v_\mathcal{I}(x_1),\cdots , v_\mathcal{I}(x_{m-1}),Z_m)$ is satisfiable. Hence, by our choice of $\rho$, $F_m(\omega_\rho(x_1),\cdots , \omega_\rho(x_{m-1}),\omega_\rho(Z_m))=\top$, or, equivalently, $\Ku_\mathcal{I},\rho\models F_m(x_1, \cdots, x_{m-1}, Z_m)$. Therefore, $\Ku_\mathcal{I},\rho\models \bigwedge_{i=1}^n \Big((x_i\wedge \neg r_i)\rightarrow F_i(x_1, \cdots, x_{i-1}, Z_i)\Big)$. 

	\item Conversely, for all $1\leq m< r$ such that $v_\mathcal{I}(x_m)=\bot$ ($m\neq r$ as, by hypothesis, $v_\mathcal{I}(x_r)=\top$), it holds that $\overline{w_{x_m}}\in\states(\rho)$. Since $m<r$, it holds that $k\geq r>m$ and $k-1\geq m\geq 1$. By the inductive hypothesis, we have that $\Ku_\mathcal{I},w_{x_m}\not\models \psi_{k-1}$. It follows that $\Ku_\mathcal{I},\overline{s_m}\cdot w_{x_m}\models \neg\psi_{k-1}\wedge \ell_{=2}$, $\Ku_\mathcal{I},\overline{w_{x_m}}\cdot\overline{s_m}\models \neg s\wedge \ell_{=2}\wedge\hsA(\neg\psi_{k-1}\wedge \ell_{=2})$ and $\Ku_\mathcal{I},\overline{w_{x_m}}\models\hsA( \neg s\wedge\ell_{=2}\wedge\hsA(\neg\psi_{k-1}\wedge \ell_{=2}))$. Hence, $\Ku_\mathcal{I},\rho\models [B]((\bigvee_{i=1}^n \hsA p_{\overline{x_i}})\rightarrow\hsA( \neg s\wedge\ell_{=2}\wedge\hsA(\neg\psi_{k-1}\wedge \ell_{=2})))$.
\end{compactitem} Combining the two cases, we can conclude that $\Ku_\mathcal{I},\rho\models \varphi_k$.
\item (Item 2) Assume that $k\geq r+1$ and $v_\mathcal{I}(x_r)=\bot$. The proof is as before and it is sketched. In this case, we choose a track $\rho=\overline{w_{x_r}}\cdots s_0$. Since $k'=k-1\geq r$, by the inductive hypothesis, $\Ku_\mathcal{I},w_{x_r}\not\models \psi_{k-1}$, and we can prove that $\Ku_\mathcal{I},\overline{w_{x_r}}\models\hsA( \neg s\wedge\ell_{=2}\wedge\hsA(\neg\psi_{k-1}\wedge \ell_{=2}))$.\qedhere
\end{compactitem}
\end{proof}

The correctness of the reduction from SNSAT to model checking for $\AB$ follows as a corollary.

\begin{corollary}\label{corol:c} Let $\mathcal{I}$ be an instance of SNSAT, with $|\mathcal{I}|=n$, and let $\Ku_\mathcal{I}$ and $\mathcal{F}_\mathcal{I}$
be defined as above. Then, $v_\mathcal{I}(x_n)=\top \iff \Ku_\mathcal{I}\models [B]\bot \rightarrow\psi_n$.
\end{corollary}
\begin{proof}
By Theorem \ref{thcorr}, $v_\mathcal{I}(x_n)=\top \iff \Ku_\mathcal{I},w_{x_n}\models \psi_n$. If $v_\mathcal{I}(x_n)=\top$, then $\Ku_\mathcal{I},w_{x_n}\models \psi_n$ and, since $w_{x_n}$ is the only initial track satisfying $[B]\bot$ (only satisfiable by tracks of length 1), $\Ku_\mathcal{I}\models [B]\bot \rightarrow\psi_n$. Conversely, if $\Ku_\mathcal{I}\models [B]\bot \rightarrow\psi_n$, then $\Ku_\mathcal{I},w_{x_n}\models \psi_n$, allowing us to conclude that $v_\mathcal{I}(x_n)=\top$.
\end{proof}

Eventually we can state the complexity of the problem.

\begin{corollary}
	The model checking problem for $\AB$ formulas over finite Kripke structures is $\PTIME^{\NP}$-hard (under $\LOGSPACE$ reductions).
\end{corollary}
\begin{proof}
	The result follows from Corollary \ref{corol:c} considering that, for an instance of SNSAT $\mathcal{I}$, with $|\mathcal{I}|=n$, $\Ku_\mathcal{I}$ and $\psi_n\in \mathcal{F}_\mathcal{I}$ have a size polynomial in $n$ and in the length of the formulas of $\mathcal{I}$. Moreover, their structures are repetitive, therefore they can be built by using logarithmic working space.
\end{proof}

We can prove the same complexity result for the symmetric fragment $\AbarE$, just by transposing the edges of $\Ku_\mathcal{I}$, and by replacing $[B]$ with $[E]$ and $\hsA$ with $\hsAt$ in the definition of $\psi_n$. 

We summarize all the $\PTIME^{\NP}$-completeness results achieved in the following statement.
\begin{corollary}
The model checking problem for $\AB$, $\AbarE$, $\AAbarB$, and $\AAbarE$ formulas over finite Kripke structures is $\PTIME^{\NP}$-complete.
\end{corollary}

We conclude the paper by providing a complexity upper and lower bound for $\AbarB$ and the symmetric fragment $\AE$. A $\Thsq$ model checking algorithm for $\AbarB$ formulas can be obtained by a suitable adaptation of the one for $\AAbar$ we devised in \cite{MMPS16}.
Due to the lack of space, we outline the construction in the appendix of~\cite{techrep}, and state here only the result. As for the hardness, we can observe that the $\Th$-hardness of $\Abar$ and $\A$, proved in~\cite{MMPS16}, immediately propagates to $\AbarB$ and $\AE$, respectively.

\begin{theorem}
The model checking problem for $\AbarB$ and $\AE$ formulas over finite Kripke structures is in $\Thsq$ and it is hard for $\Th$.
\end{theorem}

\section{Conclusions and future work}

In this paper, we have proved that the model checking problem for the HS fragments $\AB$, $\AbarE$, $\AAbarB$, and $\AAbarE$ is $\PTIME^{\NP}$-complete.
They are thus somehow ``halfway'' between $\AAbarBBbar$, $\AAbarEEbar$, and $\AAbarBbarEbar$, which are $\Psp$-complete~\cite{bmmps16,MMP15B,MMP15}, and $\HSprop$, $\B$, and $\E$, which are $\co\NP$-complete~\cite{bmmps16,MMP15B}, and $\A$, $\Abar$, and $\AAbar$, whose model checking is in $\Thsq$~\cite{MMPS16}.
%
In addition, we have shown that model checking for the HS fragments $\AbarB$ and $\AE$ has a lower complexity (it is in between $\Th$ and $\Thsq$) \cite{techrep}.

 
As for future work, we are looking for possible improvements to known complexity results for (full) HS model checking. We know that it is $\EXPSPACE$-hard
(we proved $\EXPSPACE$-hardness of its fragment $\BE$), while the only available decision procedure is nonelementary. 
We also started a comparative study of the expressiveness of HS fragments (with the current semantics as well as with some variants of it, which limit past/future branching) and of standard temporal logics, such as LTL, CTL, and CTL$^*$.

\smallskip

\noindent \textbf{Acknowledgments}.
The work by Alberto Molinari, Angelo Montanari, and Pietro Sala has been supported 
by the GNCS project \emph{Logic, Automata, and Games for Auto-Adaptive Systems}.

\bibliographystyle{eptcs}
\bibliography{bibx}



\end{document}